\documentclass[a4paper,11pt]{amsart}
\usepackage{amsfonts}
\usepackage{amsmath}
\usepackage{amssymb}
\usepackage{amsthm}
\usepackage{url}
\usepackage[english]{babel}
\usepackage{verbatim}
\usepackage{graphicx}

\usepackage{enumitem, hyperref}

\makeatletter
\def\namedlabel#1#2{\begingroup
    #2%
    \def\@currentlabel{#2}%
    \phantomsection\label{#1}\endgroup
}
\makeatother

\numberwithin{equation}{section}

\theoremstyle{plain}
\newtheorem{theorem}{Theorem}[section]
\newtheorem{lemma}[theorem]{Lemma}
\newtheorem{proposition}[theorem]{Proposition}
\newtheorem{corollary}[theorem]{Corollary}
\theoremstyle{remark}
\newtheorem{remark}{Remark}
\theoremstyle{definition}

\DeclareMathOperator*{\argmin}{arg\,min}

\begin{document}

   \author{Jun Maeda}
   \address{University of Warwick, Coventry, United Kingdom}
   \email{J.Maeda@warwick.ac.uk}
   \author{Saul D. Jacka}
   \address{University of Warwick, Coventry, United Kingdom}
   \email{S.D.Jacka@warwick.ac.uk}

  \title[A Market Driver Volatility Model via PIA]{\textit{A Market Driver Volatility Model via Policy Improvement Algorithm}}

   \begin{abstract}
In the over-the-counter market in derivatives, we sometimes see large numbers of traders taking the same position and risk. When there is this kind of concentration in the market, the position impacts the pricings of all other derivatives and changes the behaviour of the underlying volatility in a nonlinear way.\\
We model this effect using Heston's stochastic volatility model modified to take into account the impact. The impact can be incorporated into the model using a special product called a market driver, potentially with a large face value, affecting the underlying volatility itself. We derive a revised version of Heston's partial differential equation which is to be satisfied by arbitrary derivatives products in the market. This enables us to obtain valuations that reflect the actual market and helps traders identify the risks and hold appropriate assets to correctly hedge against the impact of the market driver.
\end{abstract}



   \date{\today}


\maketitle

\section{Introduction}
\label{section: Introduction}

Japan has one of the largest equity derivatives markets in the world. According to the Bank for International Settlements, Japan had \$378 billion in face value of equity-linked contracts out of the worldwide total of \$5,445 billion as of September 13, 2015  \cite{bis15}. A common underlying for equity-linked derivatives products in the country is the price-weighted Nikkei Stock Average Index (Nikkei 225) published by Nikkei Inc. Since the country is the world's third largest economy by GDP, people generally assume that the market is liquid enough to trade freely any desired position. However, in the first author's experience, this is not quite true. Long-dated volatility (especially between 2 and 5 years) is generally priced low due to the fact that the most of the traders in the market already own vega (sensitivity to volatility) from selling  (usually in significant sizes) a structured product called an 'auto-callable' to their clients. Therefore it is generally difficult to sell vega in the market. We will discuss this exotic case in more detail in a future paper, and for now, we will focus on explaining our model in a simpler context in this paper. The important fact to note is that there is a position that affects the pricings and risks of all the existing and potential derivatives products in the market.

In order to understand the background of our model, we first formulate a toy example. 

Assume that there are only 2 traders, A and B, in the over-the-counter (OTC) market. If A wants to buy volatility, A needs to buy it from B and vice versa. If A buys \$10 million of vega from B, the vega that B holds is decreased by the same amount. Generally, traders don't want to own so much risk on one side, so they might want to hedge the risk a little. Since the only market participants are A and B, they need to reverse what they previously traded in order to hedge themselves. This does not make much sense in this case, as there are only 2 market participants, but even if we assumed more participants in the market, this is still essentially what is happening: overall, the market vega is maintained and does not change whatever A and B do. Whatever A gains, B loses and vice versa.

Now introduce a new market participant C. Let us assume that C is not a participant in the OTC market but only buys vega from A and B as their client to hedge against market risk, and doesn't otherwise hedge the position (we could think of C as an insurance company, for example). If C buys \$10 million of vega from A, then A is now short the risk, so may want to buy some back in the market to hedge himself. A needs to buy it from B, of course, as C doesn't sell any vega. The important point is that the OTC market whose only participants are A and B is now short \$10 million of vega overall. The market now would like to buy some vega back. This generally drives the volatility of the underlying security or index higher. 

We elaborate this point in more detail. The demand and supply of vega could be, in general, directly converted to the supply and demand of volatility. It is easier to think of this in the Black-Scholes framework. If there is more demand for vega than supply, more people want to buy vega. The way they accomplish this is to buy plain vanilla calls and puts, which are positive vega products. If more people buy these products, the prices of the products move higher. Given other parameters are fixed, this price increase could only be explained by the increase in the underlying volatility. This is why the actual market participants refer to 'buying (selling) volatility'  when they are actually buying (selling) vega. These phrases will be used with the same meanings hereafter.

\begin{remark}
The corresponding volatility level is implied volatility as opposed to realized (or historical) volatility.
\end{remark}

Up to this point, volatility movement is just a matter of demand and supply. Now suppose that the derivative product that C bought has big second order risks, like vanna (the derivative of vega with respect to stock price) and volga (the derivative of vega with respect to volatility). For example, if the product is long vanna, vega increases when the underlying stock moves higher. In this case, A gets shorter vega just from the market movement and  he needs to buy it in the market to rehedge himself. However, if B has the same position, B gets shorter vega as well, so neither of them are interested in selling any more vega. This will make the volatility even higher. Note that in this situation, what is moving the volatility is just the change in the risk of the product that was already traded, not a new trade. We call the special product (of which the risks affect the dynamics of supply and demand of the volatility) the \textit{market driver}.

In order to model the example above, we posit a simple and easy-to-use model which is just an extension of the Heston model, one of the most popular stochastic volatility models. The core of our model is a semilinear parabolic partial differential equation (PDE) that we retrieve to price the market driver. Once we obtain the valuation of the market driver, we use a \textit{linear} parabolic PDE, which is very similar to those of Black-Scholes and Heston, to price other derivatives products.

As mentioned earlier, we are more interested in the case where the market driver is of a specific exotic type because we think its risk feedback effect is more prominent in practice. We will handle this problem in a future paper and concentrate now on the case when it is of plain vanilla type.

The problem statement so far may remind some readers of the 'feedback effect' of options which is now a somewhat mature field. The Black-Scholes model with a feedback effect models the prices of derivative products affected by delta hedging executed by program traders \cite{fs97, ps98, sp98}. It was a field which attracted a lot of interests in the 1990s. However, as far as we know, not much work has been done since then. Although the research in this paper was done separately from the studies done in the field, our ideas are very similar in the sense that some trade affects other option pricing.  We are (in a way) incorporating the feedback effect in a stochastic volatility framework. The key difference, however, is that we are not applying the feedback effect of the underlying asset (stock), but instead, \textit{that of the underlying volatility}. In the earlier models, the effect impacts the volatility passively via program traders trading the underlying asset. On the other hand, our model incorporates the effect directly in the dynamics of the volatility.

It may not look natural to incorporate a feedback effect in the volatility as it is not a tradable asset. However, from the first author's experience, supply and demand effects of the volatility do exist in the actual market and we think our model reflects, at least qualitatively, the actual market dynamics with the market driver. We believe that our model is more in line with market practitioners' perspectives than the classical feedback model.

One of the difficulties in the earlier feedback models is that they model the realized (historical) volatility rather than the implied volatility. Hedging delta of derivative products by dynamically trading the underlying asset does affect the implied volatility, but only that of short maturity. Behaviour of the current stock price has little impact over the long-dated implied volatility. Depending on the sign of the vanna of the market driver,  it is possible, for example, that even when the realized volatility increases with a large drop in the stock price, the long-dated implied volatility go lower. This cannot be modelled in the classic feedback model.

One of the benefits of our model is that the nonlinear PDEs that we derive can be approximated by a series of linear ones. The PDEs derived in the classic feedback model are generally of quasilinear type, where the nonlinearity occurs in the highest order of the equations. On the other hand, although we need a pair of PDEs, one for the market driver and the other for a general derivative product, our PDEs are at most of semilinear parabolic type, where the nonlinearity occurs in lower order of the equations. This enables us to apply a linear approximation algorithm called the Policy Improvement Algorithm (PIA) in which the approximated solution converges quickly to the actual solution of the semilinear PDE.

The reason why we introduce the PIA is that it enables us to reuse the setup for the Heston model. The Heston model has already been implemented in practice and is widely used. It is convenient to use the existing setup, whenever possible, to calculate the solutions of the new model. We also note that in the course of our research, we encountered some cases where we had a convergence in numerical solution using the PIA, but not using the finite difference method (FDM): the PIA seems to have better convergence in numerical solution than the FDM.

The rest of the paper is organized as follows:  Section~\ref{section: NewModel} explains the new model in detail. We will establish the existence and uniqueness of the solution to our PDEs in Section~\ref{section: Mathematical_Setup}. In Section~\ref{section: Control}, we transform the nonlinear PDE to an HJB equation. The PIA is then described in Section~\ref{section: PIA}. In Section~\ref{section: Numerical_Simulation}, we give a numerical example to see how valuations and risks, which are very important for day-to-day hedging for traders, change in our model from those in Heston's model. We also see in this section how PIA-approximated solutions converge to that of the nonlinear PDE. We give our conclusions in Section~\ref{section: Conclusion}.

\section{The Market Driver Model}
\label{section: NewModel}
We start by briefly reviewing Heston's stochastic volatility model \cite{hes93}.
The stochastic differential equations for the stock price and the variance are:

\begin{equation}\label{eq: Heston}
	\left\{
	\begin{array}{c}
	dS = \mu S dt + \sqrt{v} S dW^1\\
	dv = \kappa(\bar{v} - v)dt + \eta \sqrt{v} dW^2 \\
	\langle dW^1, dW^2\rangle = \rho dt \text{.}
	
\end{array}
\right.
\end{equation}

Here, $S$ denotes the underlying stock price and $v$ the variance of the underlying. $W^1$ and $W^2$ are Wiener processes with correlation $\rho$, $\mu$ is the drift of the stock, $\kappa>0$ is a constant which expresses the intensity of the mean reversion of the variance, $\bar{v}$ is the mean variance, and $\eta$ is the volatility of the variance. 

Since $v$ only takes positive values, we usually require the model to satisfy Feller's condition for avoiding the origin \cite{kt81}:

\begin{equation}\label{eq: FellerCondition}
2\kappa\bar{v}>\eta^2\text{.}
\end{equation}

With this setup, the value $V$ of a derivatives product with its satisfies Heston's PDE:

\begin{align}\label{eq: HestonPDE}
\begin{split}
\frac{\partial V}{\partial t} + r S\frac{\partial V}{\partial S} +\kappa\Big(\bar{v} - \omega v \Big)\frac{\partial V}{\partial v}+ \frac{1}{2}vS^2\frac{\partial^2 V}{\partial S^2}+ \frac{1}{2}v\eta^2\frac{\partial^2 V}{\partial v^2}+ vS\eta\rho\frac{\partial^2 V}{\partial S \partial v} - rV = 0
\end{split}
\end{align}

\noindent{with appropriate initial (or terminal, if we are calculating backwards in time) and boundary conditions. Equation \eqref{eq: HestonPDE} is a second order linear parabolic PDE. Here, $\omega$ is some constant for volatility risk premium and $r$ is the interest rate.}

Let us now assume that there is some distinguished product (called the market driver) with  value denoted by $F$. 

Using this $F$, our revised model is written as 

\begin{equation}\label{eq: Maeda_Model}
	\left\{
	\begin{array}{c}
	dS = \mu Sdt + \sqrt{v} S dW^1\\
	dv = \kappa(\bar{v} -v + Q\frac{\partial F}{\partial v})dt + \eta \sqrt{v} dW^2\\ 
	d\langle W^1, W^2 \rangle_t = \rho dt
	
\end{array}
\right.
\end{equation}

\noindent{ with some coefficient $Q$.}

Note that the only change made to the Heston SDE \eqref{eq: Heston} is the term $\kappa  Q\frac{\partial F}{\partial v}$ in the second equation. A simple justification for this is that the vega (in this paper, we use the term 'vega' for the derivative of the valuation with respect to variance, whereas it usually means the derivative of the valuation with respect to volatility) of the market driver impacts supply and demand of the variance and causes the shift in its mean. We are only adding this adjustment to the variance SDE. If we want to, we could, of course, similarly add 'delta' (derivatives of valuation with respect to the underlying stock price) adjustment in the SDE for the stock price $S$ in \eqref{eq: Maeda_Model}. However, we do not do this since i) deltas of derivatives products are generally low, so in order to have a large impact on the stock price, the face value traded on the position needs to be massive, which is not realistic and ii) the stock market is more liquid than the OTC derivatives market, in the sense that there are more people with different incentives in trading and many more people have access to the market (for example, personal investors can easily trade stocks, whereas they might need to satisfy additional requirements in order to trade derivatives. It is even harder for them to be able to trade in the OTC market due to size requirements, credit issues, and other restrictions). 

A sufficient condition for the variance not to go negative is derived by comparing the 2 processes $v$ and $v'$ starting at the same value:

\begin{equation}\label{eq: Comparison}
	\left\{
	\begin{array}{c}
	dv = \kappa(\bar{v} -v + Q\frac{\partial F}{\partial v})dt + \eta \sqrt{v} dW^2\\ 
	dv' = \kappa\{\bar{v} -v' + \min(Q\frac{\partial F}{\partial v})\}dt + \eta \sqrt{v'} dW^2\text{.}
\end{array}
\right.
\end{equation}

Since we will be working in a bounded domain, Proposition 5.2.18 in \cite{ks98} shows that $v' \le v$ almost surely. Applying Feller's condition \eqref{eq: FellerCondition} on $v'$, if

\begin{equation}\label{eq: NewFellerCondition}
2\kappa\Big\{\bar{v}+\min\Big(Q\frac{\partial F}{\partial v}\Big)\Big\}>\eta^2\text{,}
\end{equation}

\noindent{then $v' > 0$ almost surely, hence $v > 0$ almost surely. We call condition \eqref{eq: NewFellerCondition} the \textit{positive variance condition}.}

If we follow the usual argument, we obtain the following PDE for the value $V$ of a derivative:

\begin{align}\label{eq: Maeda_Model_PDE}
\begin{split}
\frac{\partial V}{\partial t} + r S\frac{\partial V}{\partial S} +\kappa\Big(\bar{v} - \omega v + Q\frac{\partial F}{\partial v}\Big)\frac{\partial V}{\partial v} + \frac{1}{2}vS^2\frac{\partial^2 V}{\partial S^2}
+ \frac{1}{2}v\eta^2\frac{\partial^2 V}{\partial v^2}+ vS\eta\rho\frac{\partial^2 V}{\partial S \partial v} - rV = 0\text{.}
\end{split}
\end{align}

Since $F$ is also the value of a specific derivative, we can substitute $V =F$ in \eqref{eq: Maeda_Model_PDE} and obtain a nonlinear PDE for $F$:

\begin{align}\label{eq: Maeda_Model_F}
\begin{split}
\frac{\partial F}{\partial t} + r S\frac{\partial F}{\partial S} + \kappa\Big(\bar{v} - \omega v + Q\frac{\partial F}{\partial v}\Big)\frac{\partial F}{\partial v} 
+ \frac{1}{2}vS^2\frac{\partial^2 F}{\partial S^2}
+ \frac{1}{2}v\eta^2\frac{\partial^2 F}{\partial v^2}+ vS\eta\rho\frac{\partial^2 F}{\partial S \partial v} - rF = 0\text{.}
\end{split}
\end{align}

Note that given $F$, the differential equation \eqref{eq: Maeda_Model_PDE} is a second order parabolic PDE that is linear in $V$ as in the Heston model. On the other hand, the differential equation \eqref{eq: Maeda_Model_F}, is semilinear. 

\section{Partial Differential Equations}
\label{section: Mathematical_Setup}

We recall some theorems from the theory of PDEs. For more detail, refer to \cite{lsu68}.

We take a bounded, open, and connected domain $\mathcal{E}$ in $\mathbb{R}_{+}^2$ which is bounded away from the axes. We further assume that $\partial \mathcal{E}$ is $C^{2+\alpha'}$ for some $\alpha'>0$. Let $Q_T = \mathcal{E}\times(0,T)$, $\mathcal{D} = \partial \mathcal{E}$, $\mathcal{D}_T = \{(x,y,t)| (x, y) \in \mathcal{D}, t\in [0,T]\}$, and $\Gamma_T = \mathcal{D}_T \cup \{(x,y,t)| (x, y)\in\mathcal{E} , t=0\}$. We impose $\psi$ as our initial and boundary conditions and assume it satisfies the compatibility condition, i.e.\ $\psi (x, y,t) \in C(\overline{Q_T})$. 

We define the differential operator $L$ by

\begin{align}
\begin{split}
-Lu: &=  r xu_x + \kappa(v_0 - \alpha y)u_y + \frac{1}{2}x^2yu_{xx}+ \frac{1}{2}\eta^2yu_{yy}+ \eta\rho{x}{y}u_{xy}\\
&= a_{ij}u_{ij} + b_iu_i
\end{split}
\end{align}

\noindent{under the Einstein summation convention.}

We reparameterize time-to-go $t$ backwards by replacing $t \rightarrow T - t$ and rewrite \eqref{eq: Maeda_Model_F} in general form:

\begin{equation}\label{eq: abstractPDE}
u_t + Lu +ru-\kappa{Q}{u_y}^2 =0\text{.}
\end{equation}

The PDE \eqref{eq: abstractPDE} is uniformly parabolic as it satisfies

\begin{equation}\label{eq: uniformelliptic}
\nu_1|\xi|^2 \le a_{ij}\xi_i\xi_j \le \nu_2 |\xi|^2 \quad \forall (x, y) \in {\overline{\mathcal{E}}} \text{,} \quad  \forall \xi \in \mathbb{R}^2
\end{equation}
\noindent{for some $\nu_1, \nu_2 > 0$.}

Theorem 6.2 of Chapter V of \cite{lsu68} shows the existence and uniqueness of the solution to \eqref{eq: abstractPDE} with continuous initial and boundary conditions. By the theorem, the solution belongs to the space $H^{\beta, \beta/2}(\overline{Q_T})$ for some $0 < \beta <1$, it also has bounded first spatial derivatives in $\overline{Q_T}$, and its second order spatial derivatives and first order time derivative belong to $H^{\gamma, \gamma/2}(\overline{Q_T})$ for some nonnegative and nonintegral number $\gamma$.

By substituting this solution in the coefficients of the PDE \eqref{eq: Maeda_Model_PDE}, Corollary 1 in Section 3.5 on page 74 of \cite{fri64} affirms the existence and uniqueness of the solution to the linear PDE for suitable initial and boundary conditions.

\begin{remark}\label{remark1}
We require the positive variance condition \eqref{eq: NewFellerCondition} to be satisfied in order to ensure that $v$ is nonnegative. Theorem 6.2 of Chapter V from \cite{lsu68} affirms that $F_y$ is bounded, but as far as the statement of the theorem goes, we don't have an explicit expression of it. For that reason, it is not easy to show that \eqref{eq: NewFellerCondition} is satisfied in general. In the case where the market driver with value $F$ is of plain vanilla type with $Q>0$, enforcing Feller's condition \eqref{eq: FellerCondition} is sufficient for the positive variance condition \eqref{eq: NewFellerCondition} to be satisfied since $\partial F/\partial y \ge 0$, and therefore $Q (\partial F/\partial y) \ge 0$.
\end{remark}

\section{Control Problem}
\label{section: Control}

From now on, we focus on solving \eqref{eq: abstractPDE}. We can apply various numerical methods, for example, the FDM, to calculate the solution numerically. If we were to do this, we would need additional resources to implement it in actual trading and in some cases, it may not be easy to do so. One of the difficulties may originate from the fact that even though it's semilinear, it's still a nonlinear PDE that we are dealing with. Applying the model to actual trading becomes more straightforward with a help of the Policy Improvement Algorithm (PIA). 

It is easy to see that \eqref{eq: abstractPDE} can be rewritten as

\begin{equation}
\inf_{\pi \in \mathbb{R}} \bigg(u_t + Lu + ru - \pi{u_y} + \frac{\pi^2}{4\kappa{Q}}\bigg) = 0\text{.}
\end{equation}

Note that this is the HJB equation to minimize

\begin{equation}\label{eq: costfunction}
V^\pi (x, y, t)= \mathbb{E}\bigg[\int^{\tau \wedge t}_0 e^{-rs}f^\pi (Z^{z, \pi}_s, t-s)ds + e^{-r(\tau\wedge t)}g(Z^{z, \pi}_{\tau\wedge t}, t\wedge\tau)\bigg]
\end{equation}

under the controlled process $Z^{z,\pi}_t := (X,Y^\pi)^T$ with SDEs

\begin{equation}\label{eq: Maeda_Model_ControlSDE}
	\left\{
	\begin{array}{c}
	dX = \mu X dt + \sqrt{Y} X dW^1\\
	dY^\pi = \kappa (\bar{v} - Y^\pi + \pi/\kappa)dt + \eta \rho\sqrt{Y^\pi}dW^1 + \eta \sqrt{Y^\pi}\sqrt{1-\rho^2} dW^2\\ 
	d\langle W^1, W^2 \rangle_t = 0
	
\end{array}
\right.
\end{equation}

with  $Z^{z,\pi}_0 =z=(x,y)^T$. Here, $f^\pi = \pi^2/4\kappa{Q}$, $g=\psi$ is the initial and boundary conditions introduced in Section~\ref{section: Mathematical_Setup}, and $\tau$ is the first hitting time of the boundary of the domain.

Our problem is now converted into the HJB equation for the following controlled initial/boundary problem:

\begin{equation}\label{eq: NonlinearIBVControl}
	\left\{
	\begin{array}{l}
	\displaystyle \inf_{\pi \in \mathbb{R}} \bigg(u_t + Lu + ru - \pi{u_y} + \frac{\pi^2}{4\kappa{Q}}\bigg) = 0 \quad (x, y, t) \in \mathcal{E} \times (0,T)\\
	u(x, y, t) = \displaystyle \inf_{\pi}V^\pi(x,y, t)\text{.}
\end{array}
\right.
\end{equation}

From the positive variance condition \eqref{eq: NewFellerCondition}, 

\begin{equation}\label{eq: FellerTypeCondition}
\pi > \frac{\eta^2}{2} - \kappa\bar{v}
\end{equation}

is sufficient for $Y$ not to go below zero.

\section{Policy Improvement Algorithm}
\label{section: PIA}

We now give a detailed formulation of the PIA and the proof of convergence. For more detail, refer to \cite{jms1} and \cite{jms2}.

Let $(\Omega, \mathcal{F},(\mathcal{F}_t)_{t\ge0},\mathbb{P})$ be a filtered probability space satisfying the usual conditions that supports a 2-dimensional $(\mathcal{F}_t)_{t\ge0}$ - Wiener process $W= (W_t)_{t\ge0}$. 

For any process $\mathcal{Y}=(\mathcal{Y}_t)_{t\ge0}$, define

\begin{equation}\label{eq: exittime}
\tau_\mathcal{E} (\mathcal{Y}):= \inf\{t \ge 0; \mathcal{Y}_t \in \partial \mathcal{E}\}\text{.}
\end{equation}

Let

\begin{align}\label{eq: controlspace}
\begin{split}
\mathcal{A}(z, T) &: = \{\Pi = (\Pi_t)_{t<T}; \Pi \text{ is adapted to } (\mathcal{F}_t)_{t<T}, \Pi_t(\omega) \in \mathbb{R} \text{ for every }t<T \text{ and }\omega \in \Omega \text{,}\\
& \text{and there exists a }\text{process }Z^{z, \Pi} \text{ that satisfies }\eqref{eq: SDE_0} \text{ and is unique in law}\}\text{,}
\end{split}
\end{align}

where

\begin{equation}\label{eq: SDE_0}
Z^{z, \Pi}_t = z + \int^t_0 \sigma(Z^{z, \Pi}_s, s, \Pi_s)dW_s + \int^t_0 \mu(Z^{z,\Pi}_s, s, \Pi_s)ds \qquad t\le T\wedge \tau_\mathcal{E}(Z^{z, \Pi})\text{.}
\end{equation}

A measurable function $\pi: \Omega \times (0,T) \rightarrow \mathbb{R}$ is a \textit{Markov policy} if for every $z \in \mathcal{E}$ and $T>0$ there exists a process $Z^{z, \pi}_t$  that is unique in law and satisfies the following:

\begin{align}\label{eq: SDE}
\begin{split}
Z^{z, \pi}_t = z + \int^t_0 &\sigma(Z^{z, \pi}_s, s, \pi(Z^{z, \pi}_s, s))dW_s + \int^t_0 \mu(Z^{z, \pi}_s, s, \pi(Z^{z, \pi}_s, s))ds \\
= z + \int^t_0& \sigma_\pi(Z^{z,\Pi}_s, s)dW_s + \int^t_0 \mu_\pi(Z^{z, \Pi}_s, s)ds \qquad t\le T\wedge \tau_\mathcal{E}(Z^{z, \pi})\text{.}
\end{split}
\end{align}

For any domain $Q_T = \mathcal{E}\times (0,T)$  and bounded measurable function $g$ defined on $\Gamma_T$, define $V^{g, \mathcal{E}, \pi}$ by

\begin{align}
\begin{split}
V^{g,\mathcal{E},\pi}(z, t) &= \mathbb{E}_{z}\bigg(\int^{t\wedge\tau}_0 e^{-rs}f^\pi (Z^{z,\pi}_s, t-s)ds+ e^{-r(t\wedge\tau)}g(Z^{z,\pi}_{t\wedge\tau}, t\wedge\tau)\bigg)\text{,}
\end{split}
\end{align}

where $f^\pi$ is the running cost and $\tau$ is the first exit time from $Q_T$.

Now define

\begin{equation}\label{eq: value_Function}
V^{g,\mathcal{E}} := \inf_{\pi\in\mathcal{A}}V^{ g, \mathcal{E},\pi}\text{.}
\end{equation}

Finally, we define the differential operator $L^\pi$:

\begin{equation}\label{eq: def_L}
L^\pi u := -u_t +\frac{1}{2}Tr\{\sigma^T_\pi (Hu) \sigma_\pi\} + \mu_\pi^T\nabla u \quad \text{ for } u\in {C}^{2,1} \text{,}
\end{equation}

where $H u$ is the Hessian of the function $u$.

\begin{proposition}\label{prop: 1}
For any Markov policy $\pi$ that is Lipschitz on compacts in $\mathbb{R}_{+}^2$, the following holds:
$V^{g, \mathcal{E}, \pi} \in {C}^{2, 1}(Q_T)$ and it satisfies
\begin{equation}\label{eq: prop1}
L^\pi V^{g,\mathcal{E}, \pi} - rV^{g,\mathcal{E}, \pi} +f^\pi= 0\text{.}
\end{equation}
\end{proposition}

\begin{proof}
It suffices to prove that $ V^{g,\mathcal{E}, \pi}$ satisfies \eqref{eq: prop1} in every open ball $U$ with $\overline{U} \subset Q_T$. 
Let $U$ be such an open ball with centre $\zeta$ and radius $\ell$. Let $z \in U$ and define $\tau$ as the first time the process $Z^{z,\pi}$ hits the boundary of $U$. For every $n \in \mathbb{N}$, define $U_n$ as the closed ball with centre $\zeta$ and radius $\ell - \frac{1}{n}$, and let $\tau_n$ be the first time the process $Z^{z,\pi}$ hits the boundary of $U_n$. 

Let $v \in {C}^{2,1}(U) \cap {C}(\overline{U})$ be the unique solution of the initial boundary value problem

\begin{equation}
	\left\{
	\begin{array}{c}
	L^\pi v - rv+f^\pi= 0\\
	v|_{\mathcal{D}_T} = V^{g, \mathcal{E}, \pi}|_{\mathcal{D}_T}\text{.}
\end{array}
\right.
\end{equation}

The existence and uniqueness is guaranteed by Corollary 1 on page 71 in \cite{fri64} and Lemma~\ref{lemma: 3} in the Appendix. The partial derivatives of $v$ are H\"{o}lder continuous by the same corollary. Let $n_0$ be large enough such that $z \in U_n$, and for every $n \ge n_0$, define the process $(J^{n})_{n\ge n_0}$ by

\begin{equation}
J^{n}_t := \int^{t\wedge\tau_{n}}_0 e^{-rs}f^\pi(Z^{z,\pi}_s, t-s)ds + e^{-r(t\wedge\tau_{n})}v(Z^{z,\pi}_{t\wedge\tau_{n}}, t-t\wedge\tau_{n})
\end{equation}

and 

\begin{equation}
J_t := \int^{t\wedge\tau}_0 e^{-rs}f^\pi(Z^{z,\pi}_s, t-s)ds + e^{-r(t\wedge\tau)}v(Z^{z,\pi}_{t\wedge\tau}, t-t\wedge\tau)\text{.}
\end{equation}

Ito's formula on $[0,\tau_n]$ and the differential equation for $v$ yield

\begin{align}
\begin{split}
J^{n}_t &= v(z, t) +\int^{t\wedge\tau_{n}}_0 e^{-rs}(f^\pi -rv +L^\pi{v})(Z^{z, \pi}_s, t-s) ds+\int^{t\wedge\tau_{n}}_0 e^{-rs}(\nabla{v})^T \sigma_\pi dW_s\\
& = v(z, t) + \int^{t\wedge\tau_{n}}_0 e^{-rs}(\nabla{v})^T \sigma_\pi dW_s\text{.}
\end{split}
\end{align}

Hence $J^{n}$ is a local martingale, and since it is clearly a bounded process, it is a uniformly integrable martingale. Thus the Dominated Convergence Theorem yields

\begin{equation}
v(z, t) =\lim_{n \rightarrow \infty}E(J^{n}_0) =\lim_{n \rightarrow \infty} E(J^{n}_{t}) = E(J_t)\text{.}
\end{equation}

From the initial and boundary conditions for $v$, we obtain

\begin{align}\label{eq: supplement_eq1}
\begin{split}
J_t &= \int^{t\wedge\tau}_0 e^{-rs}f^\pi(Z^{z,\pi}_s, t-s)ds + e^{-r(t\wedge\tau)}v(Z^{z,\pi}_{t\wedge\tau}, t-t\wedge\tau)\\
&= \int^{t\wedge\tau}_0 e^{-rs}f^\pi(Z^{z,\pi}_s,  t-s)ds + e^{-r(t\wedge\tau)}V^{ g, \mathcal{E}, \pi}(Z^{z,\pi}_{t\wedge\tau}, t-t\wedge\tau)\\
&=E\bigg(\int^{t\wedge\tau}_0 e^{-rs}f^\pi(Z^{z,\pi}_s, t-s) ds +e^{-r(t\wedge\tau)}g(Z^{z, \pi}_{t\wedge\tau}, t\wedge\tau)\bigg| \mathcal{F}_{\mathcal{S}}\bigg)\text{.}
\end{split}
\end{align}

The last equality in \eqref{eq: supplement_eq1} follows from Lemma~\ref{lemma: 4}. We conclude:

\begin{align}
\begin{split}
v(z, t) &= E\bigg(\int^{t\wedge\tau}_0 e^{-rs}f^\pi(Z^{z,\pi}_s, t-s) ds +e^{-r(t\wedge\tau)}g(Z^{z, \pi}_{t\wedge\tau}, t\wedge\tau)\bigg)\\
&= V^{ g, \mathcal{E}, \pi}(z, t)
\end{split}
\end{align}

\end{proof}

We now describe the algorithm. Let $\pi_0$ be a Markov policy that is Lipschitz on compacts in $\mathbb{R}_{+}^2$. The algorithm is defined as follows:

\begin{equation}\label{eq: inductionPDE_0}
	\left\{
	\begin{array}{c}
	\mathcal{L}u_i - \pi_{i}(u_i)_y + f^{\pi_i} =0\\
	\pi_{i+1} (z, t) = \displaystyle\argmin_{p \in A} (L^p u_i(z, t) - ru_i(z, t) + f^p(z, t))\text{,}
\end{array}
\right.
\end{equation}

where the differential operator $\mathcal{L}$ is defined as

\begin{equation}\label{eq: def_mathcal_L}
\mathcal{L} := -\frac{\partial}{\partial t} - L
\end{equation}

using $L$ in \eqref{eq: def_L}. It is important to note that $\mathcal{L}$ is independent of $\pi_i$. In our problem, \eqref{eq: inductionPDE_0} can be further calculated as

\begin{equation}\label{eq: inductionPDE}
	\left\{
	\begin{array}{c}
	\mathcal{L}u_i - \pi_{i}(u_i)_y + \pi_{i}^2/{4\kappa Q} =0\\
	\pi_{i+1}(z, t) = 2\kappa{Q}(D_y u_i)\text{,}\
\end{array}
\right.
\end{equation}

where $D_y$ denotes the partial differential operator with respect to $y$.

We already know from Section~\ref{section: Mathematical_Setup} that the solution of the semilinear PDE \eqref{eq: abstractPDE} exists uniquely with bounded spatial derivatives, so instead of $\mathcal{A}$ in \eqref{eq: controlspace}, we can take a subset of $\mathcal{A}$ of which the controls are uniformly bounded. Also, note that ${D_y}u_i$ expresses the vega of the (approximated) market driver which we assumed to be a plain vanilla. As mentioned in Section~\ref{section: Introduction}, since plain vanilla options are positive vega products and $Q>0$, we know that $\pi_i$ is nonnegative from the definition in \eqref{eq: inductionPDE}. If we assume that \eqref{eq: FellerCondition} is satisfied, then we see that the condition \eqref{eq: FellerTypeCondition} is also satisfied. This precludes $Y$ from becoming negative.

In order to apply the PIA, we need to check if the algorithm \eqref{eq: inductionPDE} satisfies the criteria of the PIA. The only criterion needed to be verified is the uniform Lipschitz condition on $\pi_i$. The following lemma proves this.

\begin{lemma}\label{lemma: convergence_lemma}
$\{\pi_i\}_i$ defined in \eqref{eq: inductionPDE} is uniformly Lipschitz continuous.
\end{lemma}

\begin{proof}
From the Schauder estimate, we have

\begin{equation}\label{eq: SchauderEstimate}
\lVert u_{i+1} \rVert_{2+\alpha} \le C(\lVert g \rVert_{2+\alpha} + \lVert f^{\pi_n} \rVert_{\alpha})\text{,}
\end{equation}

where $C$ only depends on the H{\"o}lder norms of the coefficients of $L^\pi$, the domain $Q_T$, and $\nu_1$ in \eqref{eq: uniformelliptic}.
If $g$ is continuous, we can approximate it uniformly in $2+\alpha$ norm by the  Weierstrass approximation theorem as mentioned on page 71 in \cite{fri64}. In our specific problem, $f^{\pi_n} = {\pi_n}^2/4\kappa{Q}$, so $\lVert f^{\pi_n}\rVert_{\alpha}$ is uniformly bounded thanks to the uniform boundedness of $\pi_i \in \mathcal{A}$. As the right hand side of \eqref{eq: SchauderEstimate} is uniformly bounded, $(u_i)_i$ is uniformly bounded in $2+\alpha$ norm, hence $\pi_i$ is uniformly Lipschitz continuous from the second equation in \eqref{eq: inductionPDE}.
\end{proof}

The PIA tells us that the $u_i$ in \eqref{eq: inductionPDE} converges and the limit function is $V^{g,\mathcal{E}}$ which is ${C}^{2,1}$ and satisfies the HJB equation \eqref{eq: NonlinearIBVControl} in $Q_T$.

We will see later in the actual numerical example that the convergence to the solution happens fast. In the case of a plain call option as the market driver, we get a numerical solution very close to that of the semilinear PDE with only 1 iteration.

\begin{proposition}\label{prop: 2}
$\lVert u_{i+2} - u_{i+1}\rVert_{2+\alpha} \le C\kappa{Q}\lVert u_{i+1} - u_i\rVert_{2+\alpha}^2$
\end{proposition}

\begin{proof}
By definition and Proposition~\ref{prop: 1}

\begin{equation}
	\left\{
	\begin{array}{c}
	\mathcal{L}u_{i+2} - \pi_{i+2}(u_{i+2})_y +\pi^2_{i+2}/4\kappa{Q} = 0\\
	\mathcal{L}u_{i+1} - \pi_{i+1}(u_{i+1})_y +\pi^2_{i+1}/4\kappa{Q} = 0\text{.}
\end{array}
\right.
\end{equation}

Subtracting these 2 equations and setting $v_{i+2}:=u_{i+2}- u_{i+1}$,

\begin{equation}
\mathcal{L}v_{i+2} - \pi_{i+2}(v_{i+2})_y -(\pi_{i+2} - \pi_{i+1})^2/4\kappa{Q} =0\text{.}
\end{equation}

Since $v_i$ is 0 on the parabolic boundary, from the Schauder estimate:

\begin{align}
\begin{split}
\lVert v_{i+2} \rVert_{2+\alpha} \le C\lVert \frac{\pi_{i+2} - \pi_{i+1}}{4\kappa{Q}}\rVert^2_\alpha= C\kappa{Q}\lVert (v_{i+1})_y\rVert^2_{\alpha}\le C\kappa{Q}\lVert v_{i+1}\rVert^2_{2+\alpha}
\end{split}
\end{align}

\end{proof}

Proposition~\ref{prop: 2} shows that if the approximation of the solution is close enough to the classical solution of the semilinear PDE, $\{u_i\}_i$ converges quadratically to the solution. In other words, Proposition~\ref{prop: 2} shows the quadratic local convergence of the solutions of the PIA to the classical solution.

\begin{corollary}
\begin{equation}
\lVert u_{i+1} - u_{i}\rVert_{2+\alpha} \le (C\kappa{Q}\lVert u_{1} - u_0\rVert_{2+\alpha})^{2^{i}-1}\lVert u_{1} - u_0\rVert_{2+\alpha}
\end{equation}
\end{corollary}

\begin{proof}
Use Proposition~\ref{prop: 2} and induction.
\end{proof}

\section{Numerical Simulation}
\label{section: Numerical_Simulation}

We now numerically investigate how the pricing and risks change with our model. We assume that a large amount of 2 year, 120 strike call is owned by investors outside the OTC market. We first price this structure using \eqref{eq: Maeda_Model_F}. Then, substituting this solution in \eqref{eq: Maeda_Model_PDE}, we price a different derivatives product, a 2 year, 100 strike call. We compare the results with the ones obtained from the Heston model. We used the explicit FDM method. We note that with sufficiently fine mesh in the discretization, the numerical solutions converge to the analytic ones. Hereafter, we refer to the 2 year, 120 strike call as 120 call and 2 year, 100 strike call as 100 call or at-the-money (ATM) call.

We use the parameters in Table~\ref{table: parameters}.

\begin{table}[h!]
\begin{center}
    \begin{tabular}{ | l | p{2cm} |}
    \hline
    Parameter & Value\\ \hline
    $Q$ & 0.0003 \\ \hline
    $r$ & 3.0\% \\ \hline
    $\rho$ & -0.7571 \\ \hline
    $\eta$ & 0.3 \\ \hline
    $\omega$ & 1.0 \\ \hline
    $\bar{v}$ & 0.04 \\ \hline
    $\kappa$ & 0.55 \\ \hline
    \end{tabular}
\caption{Parameters for numerical simulation.}
\label{table: parameters}
\end{center}
\end{table}

Note that Feller's condition \eqref{eq: FellerCondition} is met and $Q>0$. From Remark~\ref{remark1}, the positive variance condition \eqref{eq: NewFellerCondition} is therefore satisfied.

We take our domain $\mathcal{E}$ to be a round rectangle and denote by $S_{min}$, $S_{max}$, $v_{min}$, and $v_{max}$  the minimum and maximum values of the variables in the domain. In this example, we took $S_{min}=0.5$, $S_{max}=200$, $v_{min}=0.00005$, and $v_{max}=1.0$ and discretized each interval by 50. For the time interval $[0,2]$, we discretized it similarly by 30,000. We denote by $F_H$ the value $F$ calculated in the Heston model and by $F_N$ the value calculated in the new model. Similarly, we denote by $V_H$ and $V_N$ the corresponding values for an arbitrary $V$.

As in \cite{hes93}, for the calculation in the Heston model, we use the initial and boundary conditions:

\begin{equation}\label{eq: HestonBoundaryConditions}
	\left\{
	\begin{array}{l}
	F_H(S,v,0) = \max(0, S - K)^{+}	\qquad \qquad (S, v) \in \Omega\\
	F_H(S_{min}, v, t) = 0 \qquad \qquad \qquad \qquad \quad (S=S_{min}) \\
	\frac{\partial F_H}{\partial S}(S_{max}, v, t) = 1	\qquad \qquad \qquad  \qquad (S=S_{max})\\
	\frac{\partial F_H}{\partial t} - rS\frac{\partial F_H}{\partial S} + rF_H -\kappa\bar{v}\frac{\partial F_H}{\partial v} =0 \quad (v= v_{min})\\
	F_H(S, v_{max},t) = S \qquad \qquad \qquad  \qquad (v = v_{max})\text{.}
\end{array}
\right.
\end{equation}

The solution to the initial-boundary problem of Heston's PDE with conditions \eqref{eq: HestonBoundaryConditions} is continuous up to the boundary, so we can use the value of $F_H$ as the boundary condition for $F_N$. This way, the values of $F_H$ and $F_N$ match on the parabolic boundary.

We use corresponding boundary conditions for $V$.

With the parameters in Table~\ref{table: parameters}, the drift in the second SDE of \eqref{eq: Maeda_Model} is shifted by $\kappa{Q}\frac{\partial F_N}{\partial v}$, which in this case is calculated as $0.55\times0.0003\times77.188 = 0.0127$. This is about $58\%$ of $\kappa\overline{v}$.

The result for the 120 call (which in our case is the market driver) is shown in Table~\ref{table: 120call}.

\begin{table}[h!]
\begin{center}
    \begin{tabular}{ | l | p{1.5cm} | p{1.5cm} | p{1.5cm} | p{1.5cm} | p{1.5cm} |}
    \hline
    Risks & Value & Delta & Vega & Vanna & Volga\\ \hline
    Heston & 2.6058 & 35.378\%&70.940 & 3.8766& 119.001\\ \hline
    New Model & 3.5121 &42.457\% & 77.188& 2.4132&-535.557\\ \hline
    \end{tabular}
\caption{Summary for $120$ call at $S=98.255$ and $v = 0.030049$}
\label{table: 120call}
\end{center}
\end{table}

The result for the other derivative product (in our case, an at-the-money call) is shown in Table~\ref{table: 100call}.

\begin{table}[h!]
\begin{center}
    \begin{tabular}{ | l | p{1.5cm} | p{1.5cm} | p{1.5cm} | p{1.5cm} | p{1.5cm} |}
    \hline
    Risks & Value & Delta & Vega & Vanna & Volga\\ \hline
    Heston &11.299 & 74.117\%&79.238 & -0.5666 & -543.263\\ \hline
    New Model & 12.116& 76.942\%& 78.824& -1.6201&-961.800\\ \hline
    \end{tabular}
\caption{Summary for at-the-money (ATM) call at $S=98.255$ and $v = 0.030049$}
\label{table: 100call}
\end{center}
\end{table}

The results are for $S=98.255$ and $v = 0.030049$ at time $t=T=2$. In volatility convention (i.e.\  standard deviation, as traders usually prefer this over variance), this value of $v$ is equivalent to $\sigma = \sqrt{v} = 17.335\%$.

The obvious result is that the options are priced higher under the new model and we see it from Table~\ref{table: 120call} and Table~\ref{table: 100call}. This is due to the current set-up that the $120$ call (which is a positive vega product) is held outside of the OTC market. Since the OTC market is then overall short vega, or in other words, short volatility, the model correctly adjusts the level of the volatility which is now in demand. If we calculate the equivalent volatilities in the Heston model based on the prices we get from the new model, we get the correspondance shown in Table~\ref{table: Comparison}.

\begin{table}[h!]
\begin{center}
    \begin{tabular}{ | l | p{2cm} | p{2cm} |}
    \hline
    & $120$ Call & ATM Call\\ \hline
    Heston & 17.335\% & 17.335\% \\ \hline
    New Model & 20.694\% & 20.090\%\\ \hline
    Difference & 3.359\% & 2.755\%\\ \hline
    \end{tabular}
\caption{Implied volatility calculated based on the risk calculated in the Heston model}
\label{table: Comparison}
\end{center}
\end{table}

\noindent{From Table~\ref{table: Comparison}, we see that the volatility is higher, and the increments against the Heston volatilities are different for different structures. The result of Table~\ref{table: Comparison} shows a skewness of the impact the market driver has on the volatility.} 

To understand how large this difference in the implied volatility is, we can assume that the vega traders maintain ranges between  $\pm$\$10 million. With 3\% difference in volatility as shown in Table~\ref{table: Comparison}, if they are short \$10 million of vega, their mark-to-market loss would be -\$30 million. If their goal is to raise \$100 million of profit in a year, then this loss already corresponds to 30\% of the annual target. 

Figure~\ref{fig: stock_vega3.pdf} and Figure~\ref{fig: vol_final5.pdf} show a simulation of the processes of the stock price and the volatility.

\begin{figure}[h!]
	\centering
		\includegraphics[height=4in, width =4in]{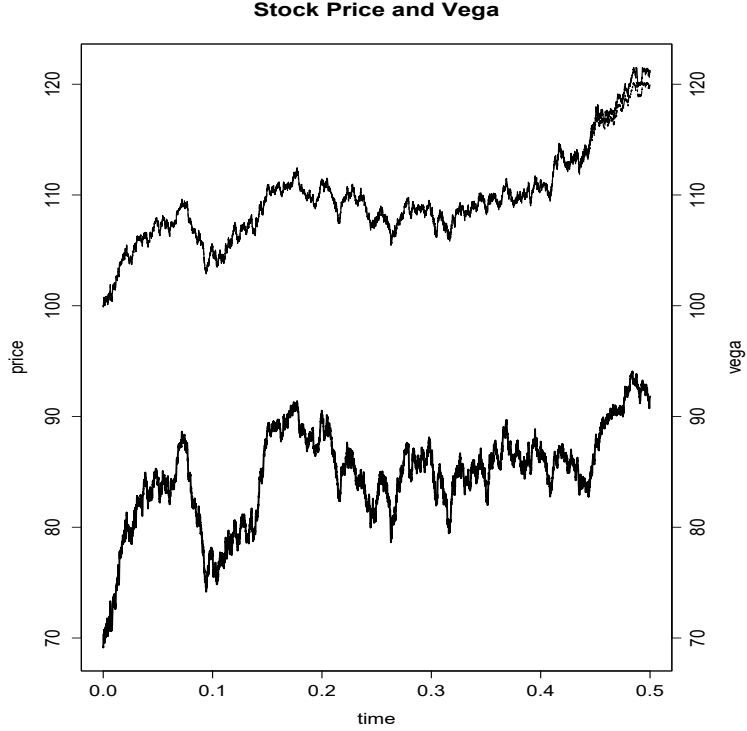}
	\caption{Simulation of the SDEs \eqref{eq: Maeda_Model} for the first 6 months starting from $S=100$ and $v = 0.04$ with $F$ being the value of the 2Y 120 call. We plotted both the stock price processes of the Heston (dotted line) and of the new model (solid line) on the upper half of the graph. The largest difference in absolute value of the realizations of the two price processes is 0.7336, which corresponds to 73.36 basis points to the initial stock price. We used the drift $\mu = 0.05$. The lower graph shows how the vega of the call in the new model changes over time.}
	\label{fig: stock_vega3.pdf}
\end{figure}

\begin{figure}[h!]
	\centering
		\includegraphics[height=4in, width =4in]{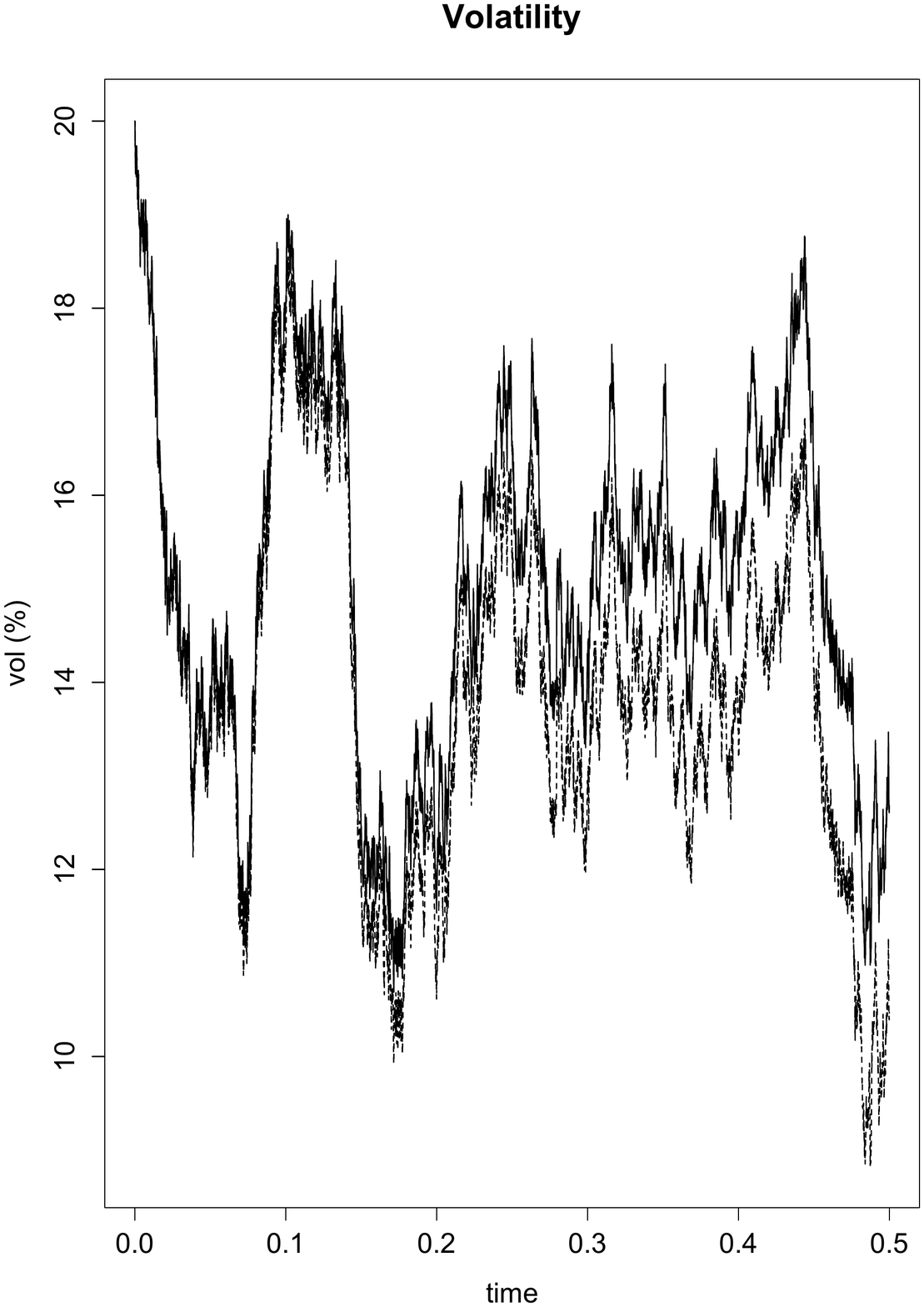}
	\caption{The volatility processes on the same simulation as in Figure~\ref{fig: stock_vega3.pdf}, where the dotted line corresponds to that of the Heston model and the solid line to that of the new model.}
	\label{fig: vol_final5.pdf}
\end{figure}

As mentioned in Section~\ref{section: Introduction}, this model prices-in not only the initial impact when some big position is traded with clients, but also the adjusted impact due to the change in the risk of the market driver. The risks change as the market moves, therefore the way traders hedge options changes under the new model. This is reflected in the graphs of the delta, vanna, and volga risks calculated in the new model compared to the ones calculated in the Heston model in Figure~\ref{fig: Graphs}. The difference in each risk is plotted in Figure~\ref{fig: Difference}. 

\begin{figure}[h]
    \centering
    \includegraphics[width=\textwidth, height=0.5\textheight]{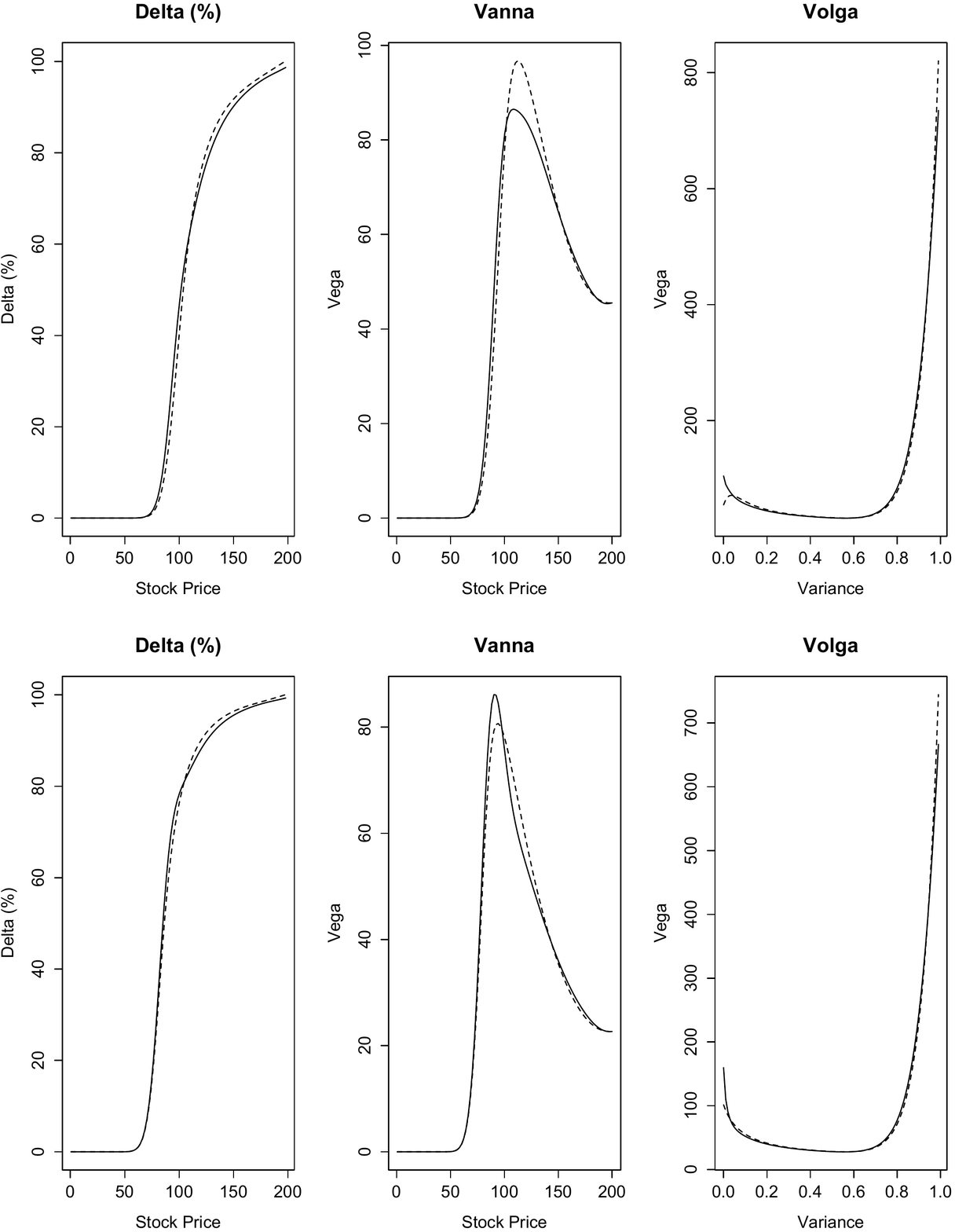}
    \caption{Risks of the calls; Top 3 charts are for 120 call and the bottom 3 are for ATM call. The solid lines indicate the risks calculated in the new model and the dotted line the corresponding risks calculated in the Heston model.}
    \label{fig: Graphs}
\end{figure}

\begin{figure}[h]
    \centering
    \includegraphics[width=\textwidth, height=0.5\textheight]{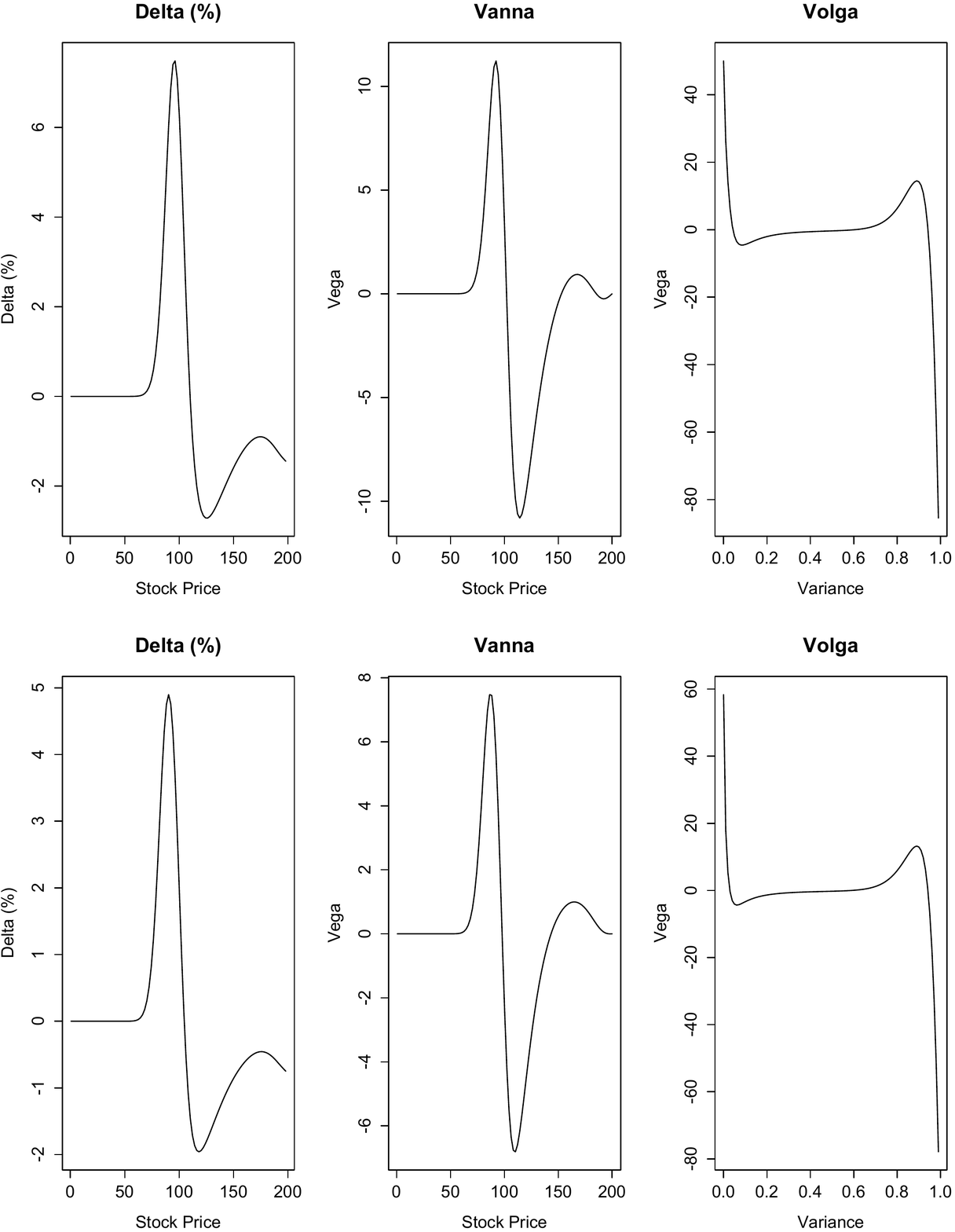}
    \caption{ The difference plotted between the values in the new model and the Heston model from Figure~\ref{fig: Graphs}.}
    \label{fig: Difference}
\end{figure}

For example, when we check the delta on Table~\ref{table: 120call} and Table~\ref{table: 100call}, the values are higher in the new model. This is because traders lose money when the stock price goes higher. To explain this in more detail, when the stock price goes higher, the vega of the $120$ call gets larger since the stock price gets closer to the strike 120. This makes the traders in the OTC market get shorter in vega, hence they will even be more eager to buy the volatility in the market. This shifts the volatility higher. The consequence of this is that the traders will lose in mark-to-market because the value of the call they are short is greater now due to the spike in volatility. The new model anticipates this and asks the traders to buy more stocks beforehand so that they are hedged from this event.

We now see what happens when we apply the PIA to the semilinear case in calculating the value of 120 call. We take $\pi_0 \equiv 0$ so that the solution of 0th iteration matches with the one from the Heston model. The result is shown in Figure~\ref{fig: PIA2.pdf}. We tried up to 4th iteration as it implies convergence in numerical solution at this point as shown in Table~\ref{table: PIADiff}.

\begin{table}[h!]
\begin{center}
    \begin{tabular}{ | l | p{1.9cm} | p{1.9cm} | p{1.9cm} | p{1.9cm} | p{1.9cm} |}
    \hline
    Iteration & 0th & 1st & 2nd & 3rd & 4th\\ \hline
    Difference & 2.3177 & 0.0322&7.72 $\times 10^{-6}$ & 0.000& 0.000\\ \hline
    \end{tabular}
\caption{Largest differences in absolute value between the numerical solutions of the approximated linear PDE and the original semilinear PDE. The figures could be regarded as the differences in percentage against the initial price of the stock as it is set to 100.}
\label{table: PIADiff}
\end{center}
\end{table}

\begin{figure}[h!]
	\centering
		\includegraphics[height=4in, width =4in]{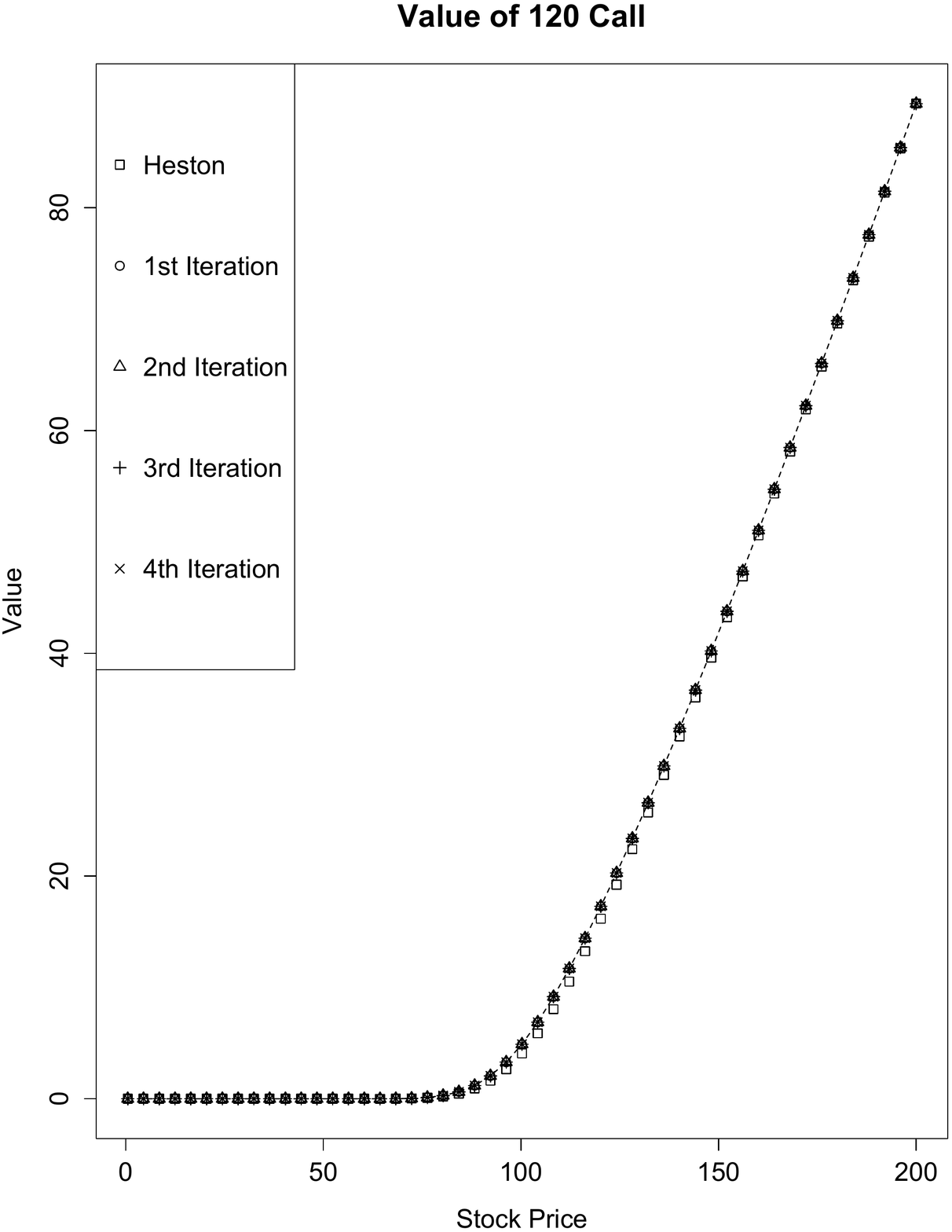}
	\caption{PIA results for 120 call. Dotted line is the solution using Finite Difference Method (FDM) directly on the semilinear PDE. $v$ is taken as $v=0.040048$.}
	\label{fig: PIA2.pdf}
\end{figure}

In Figure~\ref{fig: PIA3.pdf}, we show a magnification of Figure~\ref{fig: PIA2.pdf} centered around the stock price where we saw the largest difference, which happened to be at-the-money. 

\begin{figure}[h!]
	\centering
		\includegraphics[height=4in, width =4in]{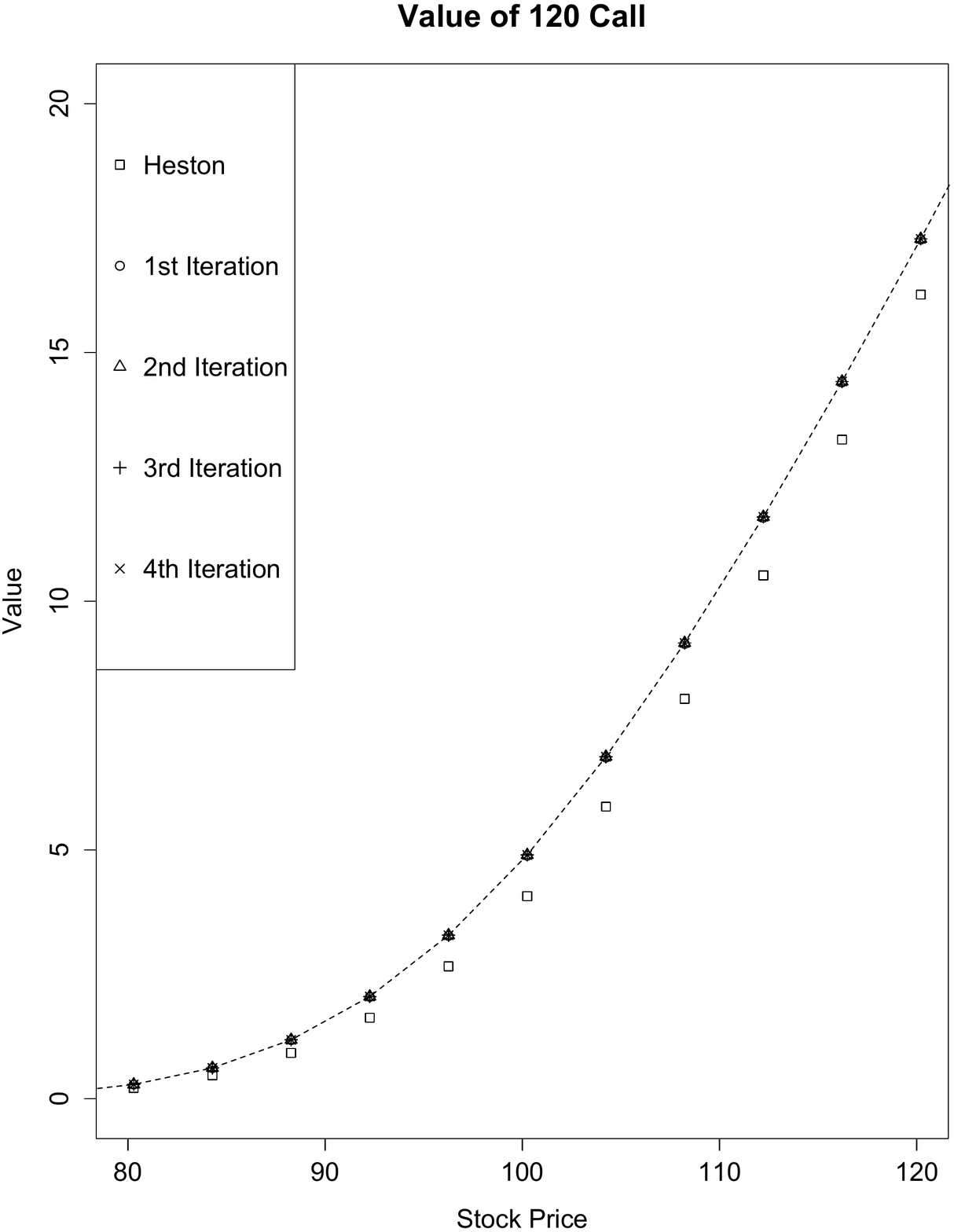}
	\caption{Magnification around at-the-money of Figure~\ref{fig: PIA2.pdf}. We see that the first iteration already approximates well the numerical solution to the semilinear PDE.}
	\label{fig: PIA3.pdf}
\end{figure}

We see in Figure~\ref{fig: PIA3.pdf} that the numerical solution of the semilinear PDE is different from that of the Heston model (0th iteration), but the 1st iteration in the PIA already brings the solution very close to that of the semilinear PDE. This is also implied by the result from Table~\ref{table: PIADiff}. This means that the numerical solution of the semilinear PDE is well approximated by a series of linear PDEs. This is good news as we don't have to create a separate program to calculate the solution to the new model, but can just reuse the same program for the Heston model with modified coefficients. The PIA also appears to have better convergence compared to the explicit FDM on  a Dirichlet boundary value problem of a second order semilinear elliptic PDE.

\section{Conclusions}
\label{section: Conclusion}

We introduced a new model which reflects the impact of a large position that is skewing the volatility market. We also introduced the Policy Improvement Algorithm. The algorithm lets us handle a semilinear PDE  as a series of linear PDEs and at the same time keep the calculation load similar to that when we run the FDM on the original semilinear problem, thanks to the fast convergence of the iterations. This enables us to easily implement the new model in practice by reusing the resources used for the Heston model which has already been widely used in the industry.

We only used a single product as a market driver, but we might try to extend this to the case when it is of a portfolio of several products. We only used a plain vanilla option as the market driver, but we should also be able to extend the model to be used for more exotic options. The difficulty then is to show the existence and uniqueness of the solution to the semilinear PDE \eqref{eq: Maeda_Model_F} and to check if the solution satisfies the positive variance condition \eqref{eq: NewFellerCondition}. If so, then by substituting this solution in the coefficient of the linear PDE \eqref{eq: Maeda_Model_PDE}, we can solve for the values of other derivatives products as in the case of the Heston model. It only takes relatively small effort to allow for the market asymmetry and to get the correct risks driven by the market driver.

The other difficulty in applying the model to actual trading appears in the calibration process. We assumed that we knew all the parameters including the detail of the market driver, but it may be challenging to recover these in the actual market, especially with more freedom in the model than in the Heston model and with limited market information.

\section*{Acknowledgements}
We would like to thank Aleksandar Mijatovi{\'c} for reading the draft and providing helpful suggestions and insights.


\renewcommand{\thesubsection}{\Alph{subsection}}
\subsection*{Appendix}
\numberwithin{figure}{subsection}
\numberwithin{equation}{subsection}
\numberwithin{theorem}{subsection}

\subsection{Lemmas}\label{App: AppendixA}

The lemmas stated here are more or less those in \cite{jms1} and \cite{jms2}. We only modify them to fit our problem. We show them here, however, so that this paper is self-contained. 

A property that forms the basis of the following lemmas is that processes controlled by Markov policies are strong Markov processes (Theorem 4.20 in \cite{ks98}). 

\begin{lemma}\label{lemma: 4}
For every Markov policy $\pi$, $z\in \mathcal{E}$, $0<t<T$, and any stopping time $\mathcal{S}$ that is almost surely less than $t \wedge\tau_\Omega$,

\begin{align}\label{eq: StrongMarkov}
\begin{split}
\mathbb{E}\bigg(\int^{t\wedge \tau}_0 e^{-rs}f^\pi(Z^{z,\pi}_s, t-s) ds +e^{-r(t\wedge\tau)}g( Z^{z,\pi}_{t\wedge \tau}, t\wedge\tau)\bigg| \mathcal{F}_\mathcal{S}\bigg)\\
=  \int^{\mathcal{S}}_0 e^{-rs}f^\pi( Z^{z,\pi}_s, t-s) ds + e^{-r\mathcal{S}}V^{g, \mathcal{E}, \pi} (&Z^{z,\pi}_\mathcal{S}, t-\mathcal{S})\text{.}
\end{split}
\end{align}

In particular, the process \\
$(\int^{T'}_0 e^{-rs}f^\pi(Z^{z,\pi}_s, t-s)ds + e^{-rT'}V^{g, \mathcal{E}, \pi}(Z^{z,\pi}_{T'}, T'))_{T'\le T}$ is a uniformly integrable martingale.
\end{lemma}

\begin{proof}
Let $\tau= \tau_\mathcal{E}(Z^{z, \pi})$ and $\tau_\mathcal{S}:=\tau \circ \theta_\mathcal{S} = \tau_\mathcal{E}(Z^{z, \pi}_{\cdot + \mathcal{S}})$, where $\theta$ is the shift operator. Then $\tau_\mathcal{S} = \tau - \mathcal{S}$ holds almost surely, and we obtain

\begin{align*}
\begin{split}
&\mathbb{E}\bigg(\int^{t\wedge \tau}_0 e^{-rs}f^\pi (Z^{z,\pi}_s, t-s)ds +e^{-r(t\wedge\tau)}g( Z^{z,\pi}_{t\wedge \tau}, t\wedge\tau) \bigg| \mathcal{F}_\mathcal{S}\bigg)\\
&= \int^{\mathcal{S}}_0 e^{-rs}f^\pi (Z^{z,\pi}_s, t-s)ds +  \mathbb{E}\bigg(\int^{t\wedge \tau}_{\mathcal{S}} e^{-rs}f^\pi (Z^{z,\pi}_s, t-s) ds +e^{-r(t\wedge\tau)}g(Z^{z,\pi}_{t\wedge \tau}, t\wedge\tau)\bigg| \mathcal{F}_\mathcal{S}\bigg)\\
&= \int^{\mathcal{S}}_0  e^{-rs}f^\pi (Z^{z,\pi}_s, t-s)ds +  \mathbb{E}\bigg(\int^{t\wedge \tau- \mathcal{S}}_{0} e^{-r(s+\mathcal{S})}f^\pi (Z^{z,\pi}_{s+\mathcal{S}}, t-(s+\mathcal{S})) ds \\
& \qquad \qquad \qquad \qquad \qquad \qquad+ e^{-r((t-\mathcal{S})\wedge(\tau-\mathcal{S}) + \mathcal{S})}g\big(Z^{z,\pi}_{(t-\mathcal{S}) \wedge(\tau - \mathcal{S}) + \mathcal{S}}, (t-\mathcal{S}) \wedge(\tau - \mathcal{S}) + \mathcal{S}, \big) \bigg| \mathcal{F}_\mathcal{S}\bigg)\\
&= \int^{\mathcal{S}}_0 e^{-rs}f^\pi (Z^{z,\pi}_s, t-s)ds +  e^{-r\mathcal{S}}\mathbb{E}\bigg(\int^{(t-\mathcal{S})\wedge \tau_\mathcal{S}}_{0} e^{-rs}f^\pi (Z^{z,\pi}_{s+\mathcal{S}}, t-(s+\mathcal{S})) ds \\
& \qquad \qquad \qquad \qquad \qquad \qquad+ e^{-r((t-\mathcal{S})\wedge\tau_\mathcal{S})}g\big(Z^{z,\pi}_{(t-\mathcal{S}) \wedge\tau_\mathcal{S}+ \mathcal{S}}, (t-\mathcal{S}) \wedge\tau_\mathcal{S} + \mathcal{S}, \big) \bigg| \mathcal{F}_\mathcal{S}\bigg)\\
&= \int^{\mathcal{S}}_0  e^{-rs}f^\pi(Z^{z,\pi}_s, t-s) ds + e^{-r\mathcal{S}} \mathbb{E}_{x}\bigg(\bigg\{\int^{(t-\mathcal{S})\wedge \tau}_{0} e^{-rs}f^\pi (Z^{z,\pi}_{s}, t-(s+\mathcal{S})) ds \\
& \qquad \qquad \qquad \qquad \qquad \qquad+e^{-r(t-\mathcal{S})\wedge\tau} g\big(Z^{z,\pi}_{ (t-\mathcal{S}) \wedge\tau}, (t-\mathcal{S}) \wedge\tau \big)\bigg\}\circ \theta_\mathcal{S} \bigg| \mathcal{F}_\mathcal{S}\bigg)\\
&= \int^{\mathcal{S}}_0  e^{-rs}f^\pi (Z^{z,\pi}_s, t-s)ds +  e^{-r\mathcal{S}}\mathbb{E}_{Z^{z,\pi}_\mathcal{S}}\bigg(\int^{(t-\mathcal{S})\wedge \tau}_{0} e^{-rs}f^\pi (Z^{z,\pi}_{s}, t-(s+\mathcal{S})) ds \\
& \qquad \qquad \qquad \qquad \qquad \qquad+ e^{-r(t-\mathcal{S})\wedge\tau}g\big(Z^{z,\pi}_{(t-\mathcal{S}) \wedge\tau}, (t-\mathcal{S}) \wedge\tau \big)\bigg)\\
&= \int^{\mathcal{S}}_0 e^{-rs} f^\pi (Z^{z,\pi}_s, t-s)ds + e^{-r\mathcal{S}} V^{g, \mathcal{E}, \pi}(Z^{z,\pi}_\mathcal{S}, t-\mathcal{S})
\end{split}
\end{align*}

\end{proof}
By taking expectation on both sides of \eqref{eq: StrongMarkov},  we retrieve a corollary which is so-called Bellman's principle.

\begin{corollary}\label{cor: 1}
For every Markov policy $\pi$, $z\in\mathcal{E}$, $0<t<T$, and stopping time $\mathcal{S}$ which is almost surely less than or equal to $t\wedge\tau_\Omega$,

\begin{align}\label{eq: Bellman}
\begin{split}
V^{g, \mathcal{E}, \pi}(z, t) =& \mathbb{E}\bigg(\int^\mathcal{S}_0 e^{-rs}f^\pi(Z^{z,\pi}_s, t-s) ds\bigg) + e^{-r\mathcal{S}}\mathbb{E}(V^{g, \mathcal{E}, \pi}( Z^{z,\pi}_\mathcal{S}, t-\mathcal{S}))\text{.}
\end{split}
\end{align}
\end{corollary}

We now use the method of mirror coupling \cite{lr86}. 

\begin{lemma}\label{lemma: 1}
For every Lipschitz Markov control and small enough $\epsilon > 0$, there exists ${\delta}>0$ such that the following holds for every $z_1, z_2 \in \mathcal{E}$: if $\lVert z_1 - z_2\rVert < {\delta}$ then there exist processes $\tilde{Z}^{z_1,\pi}$ and $\tilde{Z}^{z_2,\pi}$ that have the same laws as $Z^{z_1,\pi}$ and $Z^{z_2,\pi}$ respectively such that

\begin{equation*}
\lVert \tilde{Z}^{z_1,\pi}_t - \tilde{Z}^{z_2,\pi}_t \rVert \le G_{\tau_t} \quad on \quad {t < \rho_{{\delta}} }
\end{equation*}
and
\begin{equation*}
\tilde{Z}^{z_1,\pi}_t = \tilde{Z}^{z_2,\pi}_t \quad on \quad {t \ge \rho_0}
\end{equation*}
for every $t\ge0$, where
\begin{equation*}
\rho_c : = \inf\big\{t\ge0\text{;} \lVert \tilde{Z}^{z_1,\pi} - \tilde{Z}^{z_2,\pi} \rVert = c\big\} \quad \text{,} \quad (\inf\phi = \infty)
\end{equation*}
for any $c\ge0$, $G$ is the squared Bessel process of dimension $1+\epsilon$ started at $\lVert z_1 - z_2 \rVert$, and $(\tau_t)_{t\ge0}$ is a stochastic time change with the property
\begin{equation*}
\tau_t \le \frac{t}{\nu_1} \text{,} \quad t\ge0\text{.}
\end{equation*}
\end{lemma}

For the proof of Lemma~\ref{lemma: 1}, we refer to \cite{jms2}.

\begin{lemma}\label{lemma: 2}
For every Lipschitz Markov policy $\pi$, the function $V^{g, \mathcal{E}, \pi}(\cdot, t)$ is continuous with bounded initial condition.
\end{lemma}

\begin{proof}
Let $\epsilon > 0$ and $\hat{\delta}\le \delta$. For $\lVert z_1-z_2 \rVert \le \hat{\delta}$, we calculate $| V^{g, \mathcal{E}, \pi}(z_1, t) - V^{g, \mathcal{E}, \pi}(z_2, t) |$.

\begin{align}\label{eq: Difference}
\begin{split}
&|V^{g, \mathcal{E}, \pi}(z_1, t) - V^{g, \mathcal{E}, \pi}(z_2, t)| \\
&= \bigg| \mathbb{E}\bigg(\int^{t\wedge\tau_{z_1}}_0e^{-rs} f^\pi (\tilde{Z}^{z_1,\pi}_s, t-s)ds + e^{-r(t\wedge\tau_{z_1})}g(\tilde{Z}^{z_1,\pi}_{t\wedge\tau_{z_1}}, t\wedge\tau_{z_1})\bigg)\\
&\qquad \qquad \qquad -\mathbb{E}\bigg(\int^{t\wedge\tau_{z_2}}_0 e^{-rs}f^\pi(\tilde{Z}^{z_2,\pi}_s, t-s) ds + e^{-r(t\wedge\tau_{z_2})}g(\tilde{Z}^{z_2,\pi}_{t\wedge\tau_{z_2}}, t\wedge\tau_{z_2})\bigg)\bigg|\\
&<\bigg| \mathbb{E}\bigg(\int^{\rho_0}_0 e^{-rs}\big\{f^\pi (\tilde{Z}^{z_1,\pi}_s, t-s) -  f^\pi (\tilde{Z}^{z_2,\pi}_s, t-s)\big\}ds \bigg| I_{\rho_0 \le \rho_{\delta}}\bigg) \bigg|\\
&+\bigg| \mathbb{E}\bigg(\int^{t}_{\rho_0} e^{-rs}\big\{f^\pi (\tilde{Z}^{z_1,\pi}_s, t-s) -  f^\pi (\tilde{Z}^{z_2,\pi}_s, t-s)\big\}ds \\
& \qquad \qquad \qquad \qquad \qquad \qquad \qquad \qquad \qquad + e^{-rt}\big\{g(\tilde{Z}^{z_1,\pi}_{t}, t) -  g(\tilde{Z}^{z_2,\pi}_{t}, t) \big\}\bigg| I_{\rho_0 \le \rho_{\delta}}\bigg) \bigg|\\
&+\bigg| \mathbb{E}\bigg(\int^{t}_{0} e^{-rs}\big\{f^\pi (\tilde{Z}^{z_1,\pi}_s, t-s) -  f^\pi (\tilde{Z}^{z_2,\pi}_s, t-s)\big\}ds \\
& \qquad \qquad \qquad \qquad \qquad \qquad \qquad \qquad \qquad + e^{-rt}\big\{g(\tilde{Z}^{z_1,\pi}_{t}, t) -  g(\tilde{Z}^{z_2,\pi}_{t}, t) \big\}\bigg| I_{\rho_0 > \rho_{\delta}}\bigg) \bigg|\\
& = B_1 + B_2 + B_3\text{.}
\end{split}
\end{align}

For $B_1$, since $f^\pi$ is Lipschitz continuous, we can take $\delta_1 \in (0,\delta)$ small enough such that
\begin{equation}
B_1 < C\lVert \tilde{Z}^{z_1,\pi}_s - \tilde{Z}^{z_2,\pi}_s \rVert < \epsilon/2\text{.}
\end{equation}

For $B_2$, due to the definition of $\tilde{Z}$, the processes $\tilde{Z}^{z_1,\pi}_{t}$ and $\tilde{Z}^{z_2,\pi}_{t}$ take the same values in this time frame in consideration, so $B_2 = 0$.

Due to the boundedness of $f^\pi$ and $g$, the last term $B_3$  could be bounded by some constant multiplied by $\mathbb{P}(\rho_0 > \rho_{\delta})$. If we denote by $\rho_{\delta}(\mathcal{Y})$ and $\rho_0(\mathcal{Y})$ the first hitting times of the levels $\delta'$ and 0 respectively for any process $\mathcal{Y}$, we have from Lemma~\ref{lemma: 1}

\begin{align}
\begin{split}
\mathbb{P}(\rho_{\delta} < \rho_0) \le \mathbb{P}\bigg(\rho_{\delta}(G_\tau) < \rho_0(G_\tau)\bigg) \le  \mathbb{P}\bigg(\rho_{\delta}\bigg(G_{\frac{1}{\nu_1}}\bigg) < \rho_0\bigg(G_{\frac{1}{\nu_1}}\bigg)\bigg)\text{.}
\end{split}
\end{align}

Using the scale property of the squared Bessel process we get

\begin{align}
\begin{split}
\mathbb{P}\bigg(\rho_{\delta}\bigg(G_{\frac{1}{\nu_1}}\bigg)< \rho_0\bigg(G_{\frac{1}{\nu_1}}\bigg) \bigg)=  \mathbb{P}\bigg(\rho_{\delta}\bigg(\frac{1}{\nu_1}G\bigg) < \rho_0\bigg(\frac{1}{\nu_1}G\bigg)\bigg) =  \mathbb{P}(\rho_{\nu_1{\delta}}(G) <\rho_0 (G))\text{.}
\end{split}
\end{align}

Recall that the scale function of the Bessel process with dimension $1+\epsilon$ is given by $s(z) :=z^{\frac{1-\epsilon}{2}}$, and that the process $G$ starts at $\lVert z_1 - z_2\rVert < \hat{\delta}$. Hence we obtain

\begin{equation}
 \mathbb{P}(\rho_{\nu_1{\delta}}(G) <\rho_0 (G)) = \frac{s(\lVert z_1-z_2 \rVert) - s(0)}{s(\nu_1\delta)} \le \bigg(\frac{\hat{\delta}}{\nu_1\delta}\bigg)^{\frac{1-\epsilon}{2}}\text{.}
\end{equation}

We set $\delta = \delta_1$ and take $\hat{\delta} \in (0, \delta)$ small enough so that

\begin{equation}
2C\bigg(\frac{\hat{\delta}}{\nu_1\delta}\bigg)^{\frac{1-\epsilon}{2}} < \frac{\epsilon}{2}\text{.}
\end{equation}

Collecting what we calculated, we have proved that $\lVert z_1-z_2\rVert <\hat{\delta}$ implies $|V^{g, \mathcal{E}, \pi}(z_1, t) - V^{g, \mathcal{E}, \pi}(z_2, t)| < \epsilon$, so we have uniform continuity of $V^{g, \mathcal{E}, \pi}(\cdot, t)$.
\end{proof}

\begin{lemma}\label{lemma: 3}
For every Lipschitz Markov policy $\pi$, the function $V^{g, \mathcal{E}, \pi}$ is continuous.
\end{lemma}

\begin{proof}
If we proved the continuity of $V^{g, \mathcal{E}, \pi}$ with respect to $t$ for fixed $z$, the statement is proved using the triangle inequality and Lemma~\ref{lemma: 2}. Therefore, we prove the continuity in $t$ with fixed $z$. Due to Corollary~\ref{cor: 1}, we have

\begin{align}
\begin{split}
V^{g, \mathcal{E}, \pi}&(z, t+\delta) - V^{g, \mathcal{E}, \pi}(z, t) =\mathbb{E}\bigg(\int^{\delta}_0 e^{-rs}f^\pi(Z^{z, \pi}_s, t-s)ds\bigg)\\
& \qquad \qquad \qquad \qquad \qquad \qquad +e^{-r\delta} \mathbb{E}\bigg(V^{g, \mathcal{E}, \pi}(Z^{z, \pi}_\delta, t) - e^{r\delta}V^{g, \mathcal{E}, \pi}(z, t)\bigg)\text{.}
\end{split}
\end{align}

Applying Lemma~\ref{lemma: 2}, we obtain

\begin{equation}
|V^{g, \mathcal{E}, \pi}(z, t+\delta) - V^{g, \mathcal{E}, \pi}(z, t) | \le C\delta + C'\mathbb{E}(\lVert Z^{z, \pi}_\delta - z\rVert)\text{.}
\end{equation}

The SDE for $Z^{z, \pi}_\delta$ yields

\begin{align}
\begin{split}
Z^{z, \pi}_\delta - z = &\int^\delta_0 \mu_\pi(Z^{z, \pi}_s, s)ds + \int^\delta_0 \sigma_\pi(Z^{z, \pi}_s, s)dW_s\text{.}
\end{split}
\end{align}

Therefore, we have 

\begin{align}
\begin{split}
&|V^{g, \mathcal{E}, \pi}(z, t+\delta) - V^{g, \mathcal{E}, \pi}(z, t) | \le C\delta \\
& \qquad \qquad \qquad \qquad + C'\mathbb{E}\bigg(\bigg|\int^\delta_0 \mu_\pi(Z^{z, \pi}_s, s)ds\bigg|\bigg) + C''\mathbb{E}\bigg( \bigg|\int^\delta_0 \sigma_\pi(Z^{z, \pi}_s, s)dW_s\bigg|\bigg)\text{.}
\end{split}
\end{align}

The second term on RHS can be bounded by some multiple of $\delta$ as $\mu_\pi$ is bounded. For the last term, using Jensen's inequality and Burkholder-Davis-Gundy inequality,

\begin{align}
\begin{split}
\mathbb{E}\bigg( \bigg|\int^\delta_0 \sigma_\pi(Z^{z, \pi}_s, s)dW_s\bigg|\bigg)  &\le \bigg(\mathbb{E}\bigg(\int^\delta_0 \sigma_\pi(Z^{z, \pi}_s, s)dW_s\bigg)^2\bigg)^{\frac{1}{2}} \\
&\lesssim  \bigg(\mathbb{E}\bigg(\int^\delta_0 \sigma^2_\pi(Z^{z, \pi}_\delta, s)ds\bigg)\bigg)^{\frac{1}{2}}\text{.}
\end{split}
\end{align}

This proves the continuity of $V^{g, \mathcal{E}, \pi}$ with respect to $t$ with fixed $z$. Therefore, the continuity of $V^{g, \mathcal{E}, \pi}$ is proved.
\end{proof}


\begin{thebibliography}{99}

\bibitem{bis15}
BIS (Bank for International Settlements) 2015. Semiannual OTC Derivatives Statistics, 13 September 2015. http://www.bis.org/statistics/derstats.htm

\bibitem{duf06}
	Duffy, D.J. 2006. \emph{Finite Difference Methods in Financial Engineering: A Partial Differential Equation Approach}. {John Wiley \& Sons, Inc.}.

\bibitem{eva10}
	Evans, L.C. 2010. \emph{Partial Differential Equations}, {American Mathematical Society}, second edition.

\bibitem{fs93}
	Fleming, W.H. and Soner, H.M. 1993. \emph{Controlled Markov Processes and Viscosity Solutions}, {Springer-Verlag}.

\bibitem{fs97}
	Frey, R. and Stremme, A. 1997.  \emph{Market Volatility and Feedback Effects from Dynamic Hedging}, {Mathematical Finance}, vol. 7(4), pp. 351-374.

\bibitem{fri64}
	Friedman, A. 1964. \emph{Partial Differential Equations of Parabolic Type}, {Prentice-Hall}.

\bibitem{hes93}
	Heston, S.L. 1993.  \emph{A Closed-Form Solution for Options with Stochastic Volatility with Applications to Bond and Currency Options}, {Review of Financial Studies}, vol. 6(2), pp. 327-343.

\bibitem{hul12}
	Hull, J.C. 2012. \emph{Options, Futures, and Other Derivatives}. {Pearson}, eighth edition. 

\bibitem{iw81}
	Ikeda, N. and Watanabe, S. 1981. \emph{Stochastic Differential Equations and Diffusion Processes}, {North Holland}, {North-Holland Mathematical Library}.

\bibitem{jms1}
	Jacka, S.D., Mijatovi{\'c}, A., and {\v S}iraj, D. \emph{Policy Improvement Algorithm for Continuous Finite Horizon Problem}, forthcoming.

\bibitem{jms2}
	Jacka, S.D., Mijatovi{\'c}, A., and {\v S}iraj, D. \emph{Policy Improvement Algorithm for Controlled Multidimensional Diffusion Processes}, forthcoming.

\bibitem{ks98}
	Karatzas, I. and Shreve, S.E. 1998. \emph{Brownian Motion and Stochastic Calculus}, {Springer Science+Business Media, Inc.}, second edition.

\bibitem{kt81}
	Karlin, S. and Taylor, H.M. 1981. \emph{A Second Course in Stochastic Processes}, {Academic Press Inc.}

\bibitem{lsu68}
	Lady$\check{z}$enskaja, O.A., Solonnikov, V.A., and Ural'ceva, N.N. 1968. \emph{Linear and Quasi-linear Equations of Parabolic Type}, {American Mathematical Society}.

\bibitem{lie96}
	Lieberman, G.M. 1996. \emph{Second Order Parabolic Differential Equations}. {World Scientific Publishing Co. Pte. Ltd.}

\bibitem{lin08}
	Lin, S. 2008. \emph{Finite Difference Schemes for Heston Model}. {Master's thesis}.{University of Oxford}.

\bibitem{lr86}
	Lindvall, T. and Rogers, L.C.G. 1986. \emph{Coupling of Multidimensional Diffusions by Reflection}, {The Annals of Probability}, vol. 14(3), pp. 860-872.

\bibitem{lkd10}
	Lord, R., Koekkoek, R., and Dijk, D. 2010. \emph{A Comparison of Biased Simulation Schemes for Stochastic Volatility Models}, {Quantitative Finance}, vol. 10(2), pp. 177-194.

\bibitem{ps98}
	Platen, E. and Schweizer, M. 1998. \emph{On Feedback Effects from Hedging Derivatives}, {Mathematical Finance}, vol. 8(1), pp. 67-84.

\bibitem{roc70}
	Rockafellar, R.T. 1970. \emph{Convex Analysis}, {Princeton University Press}.

\bibitem{sp98}
	Sircar, K.R. and Papanicolaou, G. 1998. \emph{General Black-Scholes Models Accounting for Increased Market Volatility from Hedging Strategies}, {Applied Mathematical Finance}, vol. 5(1), pp. 45-82.

\bibitem{smi85}
	Smith, G.D. 1985. \emph{Numerical Solution of Partial Differential Equations: Finite Difference Methods}. {Oxford Applied Mathematics and Computing Science Series}. {Oxford}, third edition. 

\bibitem{son07}
	Soner, H. 2007. \emph{Stochastic Representations for Nonlinear Parabolic PDEs} in \emph{Handbook of Differential Equations}, {Evolutionary Equations}, vol. 3, Chapter 6, Elsevier B.V.

\bibitem{tr00}
	Tavella, D., and Randall, C. 2000. \emph{Pricing Financial Instruments: The Finite Difference Method}.
{John Wiley \& Sons, Inc.}
\end{thebibliography}
\end{document}